\documentclass{article}

\usepackage[preprint]{neurips_2026}

\usepackage{booktabs}
\usepackage{tabularx}
\usepackage{microtype}
\usepackage{graphicx}
\usepackage{subcaption}
\usepackage{booktabs} 
\usepackage{xcolor}
\usepackage{hyperref}

\usepackage{amsmath}
\usepackage{amssymb}
\usepackage{mathtools}
\usepackage{amsthm}
\usepackage{algorithm,algorithmic}

\usepackage[capitalize,noabbrev]{cleveref}

\theoremstyle{plain}
\newtheorem{theorem}{Theorem}[section]

\newtheorem{lemma}[theorem]{Lemma}
\newtheorem{corollary}[theorem]{Corollary}
\theoremstyle{definition}
\newtheorem{definition}[theorem]{Definition}

\theoremstyle{remark}

\DeclareMathOperator{\sinc}{sinc}

\title{Spectral Filtering for \\ Complex Linear Dynamical Systems}
\author{%
  Elad Hazan $^\dagger$ \thanks{Princeton University}
  \And
  Annie Marsden 
  \thanks{Google DeepMind} 
}
\begin{document}
\maketitle

\newcommand{\Cbeta}{\mathbb{C}_\beta}
\newcommand{\ZW}{Z_W(\beta)}
\newcommand{\y}{\mathbf{y}}
\newcommand{\h}{\mathbf{h}}
\newcommand{\norm}[1]{\lvert \lvert #1 \rvert \rvert}

\begin{abstract}

We study the problem of learning complex-valued linear dynamical systems (CLDS) with sector-bounded spectrum. This class captures oscillatory and long-memory dynamics arising in signal processing, structured state space models, and quantum systems. We introduce a spectral filtering method based on the Slepian basis and show that learnability is governed by an effective dimension
independent of the ambient state dimension. As a consequence, we obtain dimension-free regret bounds for sequence prediction in CLDS with spectrum contained in a sector of the unit disk.

\end{abstract}

\section{Introduction}

Learning dynamical systems with complex-valued spectra is a central problem in modern machine learning and control. Such systems naturally arise in settings involving oscillatory or long-range temporal structure, including signal processing, linear recurrent neural networks, structured state space models (SSMs) \cite{gu2022s4,gupta2022dss,fu2023h3,gu2024mamba}, and quantum dynamics. In these applications, complex eigenvalues enable compact representations of phase and frequency information, but they also induce slow decay and long memory, making learning and prediction challenging.

A standard approach to learning dynamical systems is system identification, which aims to recover the underlying state-space parameters. However, for high-dimensional systems, this approach is often computationally prohibitive and statistically inefficient, as its complexity scales with the ambient state dimension. An alternative paradigm, which has gained traction in recent years, is \emph{improper dynamic learning}, where the goal is to predict future observations directly from past data without explicitly reconstructing the system parameters.

\subsection{Problem Setup}

Our goal is to sequentially predict a stream of observations $y_1,y_2,\dots$ from past inputs and outputs, under squared loss,
\[
\sum_{t=1}^T \|y_t-\hat y_t\|^2.
\]
We assume the observations are generated by a complex-valued linear dynamical system (CLDS) of the form
\begin{equation}
\label{eq:LDS_connection}
    y_t = \sum_{\tau=1}^{\infty} C A^{\tau-1} B u_{t-\tau},
\end{equation}
where $A \in \mathbb{C}^{d \times d}$ is the hidden transition matrix, $B$ is the input map, and $C$ is the observation map. The hidden dimension $d$ may be very large, and our objective is to design learning algorithms whose guarantees depend only mildly---ideally logarithmically---on this ambient dimension.

In principle, a stable linear dynamical system may have an infinite impulse response, so that outputs depend on the entire past input sequence. In modern sequence prediction, however, both algorithms and deployed models operate with a finite context window: at prediction time, only the most recent $W$ inputs are retained and processed. We therefore formulate the learning problem over a finite prediction window $W$. Equivalently, the learner competes with the best sector-bounded CLDS predictor after truncating its impulse response to the most recent W lags:
\[
y_t = \sum_{\tau=1}^{W} C A^{\tau-1} B u_{t-\tau}.
\]
Thus, the prediction problem is to learn from a window of the previous $W$ inputs while competing with the best predictor induced by the underlying CLDS. 

We impose a spectral assumption on the transition matrix $A$. Specifically, we assume that its spectrum lies in the sector
\begin{equation}
    \Cbeta \triangleq \{ z \in \mathbb{C} : |z| \le 1,\ |\arg(z)| \le \beta \}.
\end{equation}
This condition controls stability and oscillatory complexity of the dynamics through the angular aperture $\beta$. Intuitively, smaller $\beta$ corresponds to lower effective frequency content, and hence to a smaller intrinsic prediction complexity.

This model also has potential relevance to quantum dynamics: in linear-response or reduced-order regimes, driven quantum systems are often described by complex linear evolutions, for example through Liouvillian generators together with control and measurement operators.

\subsection{Main Results}

We give an efficient algorithm, described in Algorithm~\ref{alg:qsf}, based on a new variant of spectral filtering for complex-valued linear dynamics. Our main result shows that the prediction error depends on the \emph{effective spectral complexity} of the system, rather than on the hidden state dimension.

Informally, let $\{y_t\}_{t=1}^T$ be generated by a noiseless CLDS whose transition matrix has spectrum contained in $\mathbb{C}_\beta$ and whose memory is bounded by $W$. Then Algorithm~\ref{alg:qsf} produces causal predictions $\{\hat y_t\}$ satisfying
\[
\frac{1}{T}\sum_{t=1}^T |\hat y_t-y_t|^2
    \le
    \tilde O\!\left(\frac{\beta W \log^2 T}{T}\right).
\]
The precise statement appears in Corollary~\ref{cor:main}.

The key qualitative features of this guarantee are as follows.
\begin{enumerate}
    \item \textbf{Dimension-free learning.}  
    The running time and prediction guarantee do not depend on the hidden dimension $d$ of the dynamical system. Instead, the relevant complexity is an effective dimension governed by the sector width $\beta$ and the memory window $W$. Our analysis shows that this reduction is a manifestation of the classical Shannon-number phenomenon from time--frequency concentration, via the spectral concentration of discrete prolate spheroidal sequences (DPSS) \cite{slepian1978prolate}.

    \item \textbf{Compression of long memory.}  
    Although the system has memory horizon $W$, the number of learned parameters is controlled by the effective dimension rather than by $W$ itself. 
Moreover, fast convolutional implementations of spectral filtering can significantly reduce the dependence on $W$, following ideas from \cite{agarwal2024futurefill}.
    
\end{enumerate}

For comparison, a direct regression approach over the previous $W$ observations would yield the bound
\[
\frac{1}{T}\sum_{t=1}^T |\hat y_t-y_t|^2
    \le
    \tilde O\!\left(\frac{W \log T}{T}\right),
\]
using standard online linear regression methods \cite{CesaBianchi2006,Hazan2022}. However, such an approach requires $\Theta(W)$ parameters and memory, and its running time is at least quadratic in $W$. Alternatively, the Cayley-Hamilton theorem implies that a linear filter over the past observations can accurately predict a CLDS with $d$ coefficients, where $d$ is the hidden dimension of the system, see e.g. \cite{hazan2022introduction} chapter 9. This means that both the running time and regret would depend on $d$. 

The comparison is summarized in Table~\ref{tab:comparison}: our guarantee
replaces the raw window length \(W\) and the hidden dimension \(d\) by the
sector-controlled effective dimension \(k^\star\simeq \beta W/\pi\).

\begin{table}[h!]
\centering
\small
\setlength{\tabcolsep}{3.5pt}
\renewcommand{\arraystretch}{1.35}

\begin{tabular}{|
>{\centering\arraybackslash}p{0.20\linewidth}|
>{\centering\arraybackslash}p{0.14\linewidth}|
>{\centering\arraybackslash}p{0.17\linewidth}|
>{\raggedright\arraybackslash}p{0.38\linewidth}|}
\hline
\textbf{Average error (omitted log factors)} 
& \textbf{Dependence on $d$} 
& \textbf{Complex angle} 
& \textbf{Methods} \\ 
\hline

${W}/{T}$ 
& $O(1)$ 
& $\pi$ 
& Window regression over the last $W$ inputs. \\
\hline

$ d/T$ 
& $d$ 
& $\pi$ 
& Closed-loop / Cayley--Hamilton regression. \\
\hline

$ T^{-1/2}$ 
& $O(1)$ 
& $0$ (real)
& Spectral filtering 
\citep{hazan2017learning}. \\
\hline

$ T^{-2/13}$ 
& $O(1)$ 
& $\log^{-2} T$ 
& USP + regression \citep{marsden2025universal}. \\
\hline

$ \beta W/T$ 
& $O(1)$
& $\beta$ 
& \textbf{Ours.} \\
\hline
\end{tabular}

\caption{\label{tab:comparison}
Schematic comparison of average prediction error bounds for marginally stable
LDS/CLDS, omitting logarithmic and problem-dependent factors. The losses and
assumptions are not identical across rows: spectral filtering applies to
real/symmetric dynamics, \citep{marsden2025universal} is stated under complex-angle
assumptions, and our result is a realizable finite-window squared-loss bound.
The ``complex angle'' column denotes the required upper bound on
\(\max_j |\arg(\lambda_j)|\). Our guarantee is finite-window: in cumulative
form it is \(\widetilde O(\beta W)\), or equivalently
\(\widetilde O(\beta W/T)\) average error.
}
\end{table}

We also prove in Appendix~\ref{appendix:lower_bound} that this improvement fundamentally relies on the spectral restriction. In particular, when no bound on the sector width is imposed, one cannot in general avoid an $\Omega(W/T)$ dependence, even information-theoretically.

\section{Related Work}

\paragraph{Learning complex linear dynamical systems.}
Learning linear dynamical systems with asymmetric transition operators and complex eigenvalues is a classical challenge in prediction and control. Complex spectra naturally encode oscillatory behavior and long-range temporal structure, but they also make both identification and online prediction significantly harder. Existing methods for general LDS and recurrent linear models can handle such systems, but their guarantees typically scale at least linearly with the hidden state dimension \cite{ghai2020no,bakshi2023new}. This dependence is prohibitive in the high-dimensional settings that motivate modern sequence modeling. In contrast, our goal is to exploit spectral structure so that both computation and learning guarantees depend on an intrinsic complexity measure rather than on the ambient dimension.

\paragraph{Improper learning and spectral filtering.}
Our work is most closely related to the literature on \emph{improper learning} of dynamical systems, where one predicts future observations directly from filtered histories without explicitly recovering the underlying state-space parameters. This viewpoint has led to algorithms with sublinear regret for broad classes of LDS, and in some cases near-dimension-free guarantees \cite{marsden2025universal}. Our contribution is in this spirit, but tailored to complex-valued systems whose spectrum lies in a sector $\mathbb{C}_\beta$ of the unit disk. By constructing filters from the Slepian/DPSS basis matched to this spectral geometry, we obtain guarantees governed by an effective dimension of order $\beta W$, where $W$ is the memory horizon. A related result of \cite{franccois2025uncertainty} studies a much more restricted shift-based setting and derives sharp retrieval limits; our results may be viewed as extending this type of information cutoff phenomenon to a substantially broader class of complex linear systems.

\paragraph{Structured state space models and complex parameterizations.}
Recent advances in structured state space models (SSMs), including S4 and its diagonal variants, have highlighted the importance of spectral parameterizations for efficient long-range sequence modeling~\cite{gu2022s4,gupta2022dss,fu2023h3,gu2024mamba}. Our perspective is further motivated by theoretical work showing that complex eigenvalues are not merely a convenient modeling choice, but are often essential for efficiently representing oscillatory and long-memory dynamics. In particular, \cite{ran2024provable} establish strong limitations for real-valued diagonal SSMs when approximating oscillatory systems, while \cite{orvieto2023universality} identify conditioning barriers that arise when complex modes are excluded. These results support the central structural assumption in our work: for oscillatory sequence prediction, the relevant notion of complexity is often spectral localization in the complex plane rather than low hidden dimension.

\paragraph{Time--frequency concentration and Slepian bases.}
The mathematical mechanism behind our compression result comes from the classical theory of time--frequency concentration. Discrete prolate spheroidal sequences (DPSS), introduced by Slepian and collaborators, form an optimally concentrated basis for sequences that are simultaneously time-limited and band-limited \cite{slepian1978prolate}. We leverage the same principle in the setting of complex linear dynamics. When the spectrum of the transition operator is confined to a sector $\mathbb{C}_\beta$, the resulting impulse responses have an intrinsic bandwidth, and the associated predictor class admits a low-dimensional approximation governed by a Shannon-number-type quantity. This connection underlies both our spectral analysis and the design of our filtering algorithm.

\paragraph{Prediction versus identification.}
Classical system identification aims to recover the latent parameters $(A,B,C)$ of the underlying state-space model \cite{simchowitz2019semiparametric,tsiamis2019stochsysid,lee2022improved,bakshi2023new}. While this objective is appropriate when one seeks interpretability or control design, it is often unnecessarily strong for forecasting. In high-dimensional settings, exact recovery can be statistically inefficient and computationally prohibitive even when accurate prediction is possible. Our work therefore adopts the improper-learning viewpoint: we compete directly with the predictor induced by the unknown system rather than attempting to reconstruct the system itself.


\section{Spectral Geometry of Complex Linear Dynamics}

The central structural question in this paper is the following: how large is the class of length-$W$ responses generated by a complex linear dynamical system whose hidden spectrum lies in the sector
$ \Cbeta \triangleq \{ z \in \mathbb{C} : |z| \le 1,\ |\arg(z)| \le \beta \} $?
Although such responses live in the ambient space $\mathbb{C}^W$, we will show that their effective complexity is much smaller, and is governed by the spectral aperture $\beta$ and the memory horizon $W$ rather than by the hidden state dimension.

We first introduce the matrix that captures this geometry and shows that its singular values exhibit a sharp cutoff. That cutoff will later determine the number of features needed for efficient prediction.

\paragraph{The Spectral Information Matrix.}
For intuition, and exactly in the diagonalizable case, a single eigenmode with eigenvalue $z \in \Cbeta$ contributes a truncated impulse response of the form
\[
\mu_W(z) \triangleq [1, z, z^2, \dots, z^{W-1}]^\top \in \mathbb{C}^W.
\]
These monomial trajectories are the basic building blocks of the predictor class induced by a sector-bounded CLDS. This motivates the following averaged Gram matrix.

\begin{definition}[Spectral Information Matrix]
\label{def:sim}
The \emph{Spectral Information Matrix} associated with window length $W$ and sector width $\beta$ is
\begin{equation}
\label{eq:integral_def}
    Z_W(\beta)
    \triangleq
    \int_{\Cbeta} \mu_W(z)\mu_W(z)^\dagger \, dA(z),
\end{equation}
where $dA(z)=r\,dr\,d\theta$ denotes area measure in polar coordinates.
\end{definition}

Evaluating the integral in polar coordinates, for $0 \le j,k \le W-1$ we obtain
\begin{align}
    (Z_W(\beta))_{jk}
    &=
    \int_{-\beta}^{\beta}\int_0^1
    r^{j+k} e^{i(j-k)\theta}\, r\,dr\,d\theta \notag\\
    &=
    \left(\int_0^1 r^{j+k+1}\,dr\right)
    \left(\int_{-\beta}^{\beta} e^{i(j-k)\theta}\,d\theta\right) \notag\\
    &=
    \frac{2\beta}{j+k+2}\,\sinc\!\big((j-k)\beta\big).
\end{align}

This formula reveals a clean separation between radial decay and angular oscillation:
\begin{equation}
    Z_W(\beta)=H_W \circ S_W(\beta),
\end{equation}
where $\circ$ denotes the Hadamard product, and
\[
(H_W)_{jk}=\frac{1}{j+k+2},
\qquad
(S_W(\beta))_{jk}=2\beta\,\sinc\!\big((j-k)\beta\big).
\]

The two factors have distinct interpretations:
\begin{enumerate}
    \item \textbf{Radial/stability factor.}
    $H_W$ is a Hankel matrix that depends only on $j+k$. It captures the effect of radial decay $|z|<1$, i.e.\ how quickly past inputs are attenuated.
    \item \textbf{Angular/oscillatory factor.}
    $S_W(\beta)$ is a Toeplitz matrix that depends only on $j-k$. It captures the phase variation allowed by the sector and coincides with the finite-window Slepian/DPSS concentration kernel \cite{slepian1978prolate}.
\end{enumerate}

\subsection{Effective Spectral Dimension and Information Cutoff}

A basic consequence of the representation above is that the class of admissible length-$W$ responses has an intrinsic dimension much smaller than $W$. 
\begin{definition}[Effective Spectral Dimension]
\label{def:effective_dimension}
The \emph{effective spectral dimension} associated with horizon $W$ and sector width $\beta$ is
\[
k^\star \triangleq \left\lceil \frac{\beta}{\pi}W \right\rceil.
\]
\end{definition}

The quantity $k^\star$ plays the role of a Shannon number for sector-bounded CLDS. It measures the number of oscillatory degrees of freedom that can be resolved over a horizon of length $W$ when the spectrum is confined to $\Cbeta$.

The next theorem formalizes this intuition.
\begin{theorem}[Information Cutoff]
\label{thm:cutoff}
Let \(Z_W(\beta)\in \mathbb{C}^{W\times W}\) be the Spectral Information Matrix, and let
$
k^\star \;:=\; \left\lceil \frac{\beta}{\pi} W \right\rceil .
$
Then for each fixed \(\beta\in(0,\pi]\), there exist constants
\(C_\beta,c_\beta>0\) such that for all \(k \ge k^\star\),
\[
\sigma_k\!\left(Z_W(\beta)\right)
\;\le\;
C_\beta
\exp\!\left(
-\,c_\beta\,\frac{k}{k^\star \log(W)}
\right).
\]
\end{theorem}

In particular, beyond the effective dimension \(k^\star\), the singular values of
\(Z_W(\beta)\) decay at a stretched-exponential rate, with only an additional
\(\log W\) loss in the exponent.

This theorem shows that the informative part of the response class is effectively confined to a $k^\star$-dimensional subspace. Beyond that threshold, the singular values of $Z_W(\beta)$ decay exponentially, so the remaining directions contribute only exponentially small information.

Equivalently, although a CLDS may have a very large hidden state dimension, the set of predictors it induces over a finite window can be compressed to dimension on the order of $\beta W$.

\begin{proof}[Proof Sketch]
The matrix $Z_W(\beta)$ factors as the Hadamard product
\[
Z_W(\beta)=H_W \circ S_W(\beta),
\]
where $H_W$ is a Hankel matrix with rapidly decaying spectrum and $S_W(\beta)$ is the finite-window Slepian concentration matrix. Classical DPSS theory implies that $S_W(\beta)$ has roughly $k^\star\approx \frac{\beta}{\pi}W$ significant eigenvalues, after which its spectrum plunges. On the other hand, the singular values of $H_W$ decay quickly because of its smoothing Hankel structure.

A Hadamard-product tail bound then shows that truncating $H_W$ to rank $s$ and $S_W(\beta)$ to rank $k^\star$ yields an approximation to $Z_W(\beta)$ of rank at most $s k^\star$, with error controlled by the discarded tails of the two factors. Combining these two ingredients yields exponential decay beyond the effective dimension. The full proof appears in Appendix~\ref{app:hadamard_decay}.
\end{proof}

Figure~\ref{fig:spectrum} illustrates this phenomenon empirically: the singular spectrum exhibits a sharp transition around $k^\star$, after which the tail rapidly collapses.

\begin{figure}[h!]
    \centering
    \includegraphics[width=0.8\linewidth]{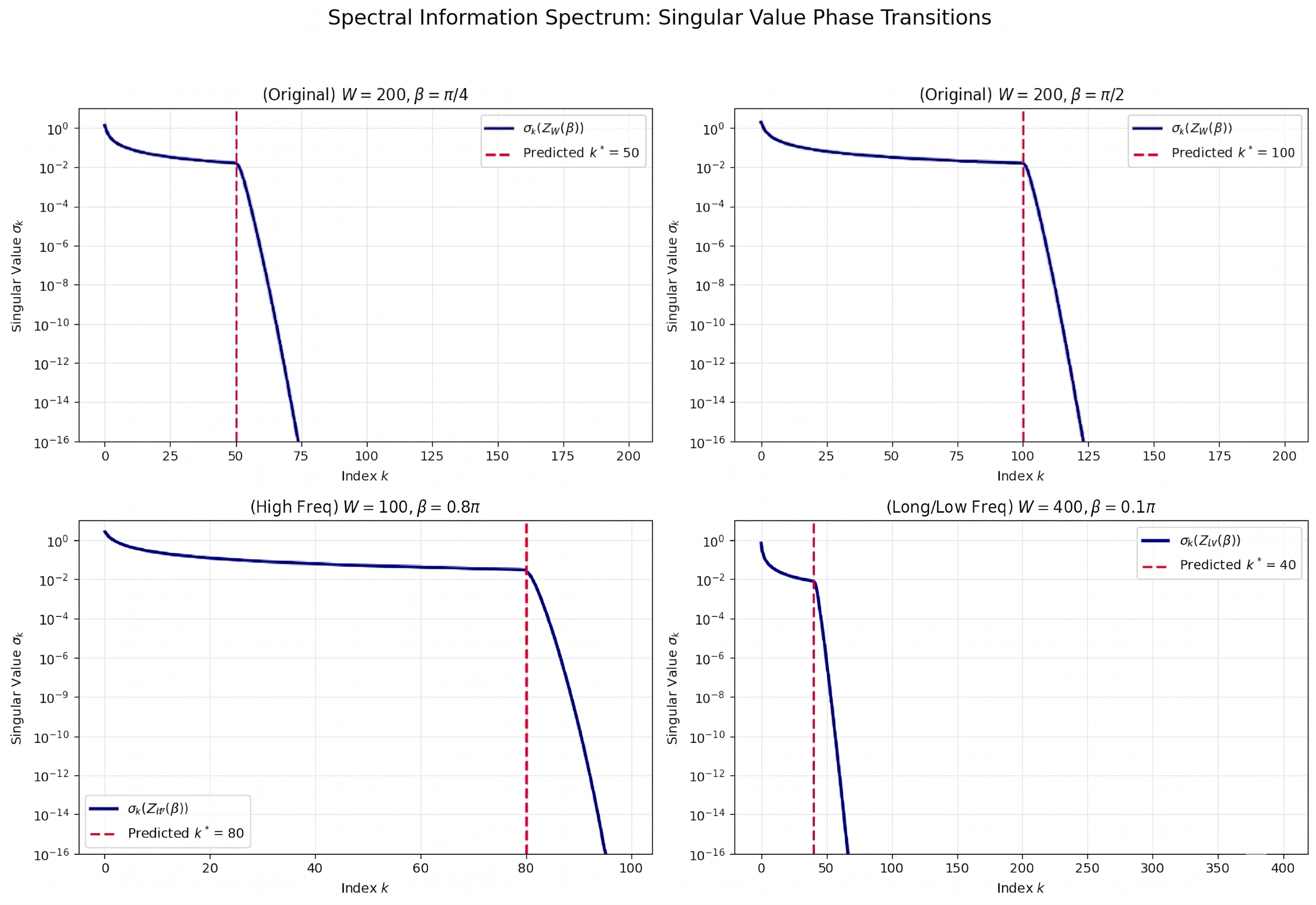}
    \caption{\textbf{Spectral Information Spectrum.} Singular values of $Z_W(\beta)$ on a log scale. The red vertical lines indicate the predicted cutoff
    $k^\star=\left\lceil \frac{\beta}{\pi}W \right\rceil$, separating the informative part of the spectrum from the rapidly decaying tail.}
    \label{fig:spectrum}
\end{figure}

\section{Algorithm: Complex Spectral Filtering}

Theorem~\ref{thm:cutoff} suggests a universal improper-learning strategy for sector-bounded CLDS. Rather than estimating the unknown state-space parameters $(A,B,C)$, we first construct a basis adapted to the class of admissible impulse responses, and then learn only a predictor in that compressed feature space.

This leads to the \emph{Complex Spectral Filtering} algorithm. Its key advantage is that both the feature construction and the learning problem depend on the effective spectral dimension of the class, not on the hidden dimension of the underlying system.

\subsection{Universal Filter Construction}

Let
\[
Z_W(\beta) = \sum_{i=1}^W \sigma_i v_i v_i^\dagger
\]
be the eigendecomposition of the Spectral Information Matrix, with
$
\sigma_1 \ge \sigma_2 \ge \cdots \ge \sigma_W.
$
For a chosen filter budget $k$, we define the universal filter bank
\begin{equation}
    \Phi \triangleq [v_1,\dots,v_k] \in \mathbb{C}^{W\times k}.
\end{equation}

By Theorem~\ref{thm:cutoff}, when $k$ is on the order of $k^\star$ (or modestly larger), the discarded tail contributes only exponentially small approximation error. Thus, the columns of $\Phi$ provide a near-optimal basis for compressing the history of any CLDS whose spectrum lies in $\Cbeta$.

A crucial point is that the filter bank is \emph{universal}: it depends only on the pair $(W,\beta)$, not on the unknown realization $(A,B,C)$ of the system. In other words, the geometry of the model class determines the representation.

\subsection{Learning via Filtered Histories}

At time $t$, we form the length-$W$ input history
\[
U_t \triangleq [u_t,u_{t-1},\dots,u_{t-W+1}]^\top \in \mathbb{C}^W,
\]
with zero-padding when $t<W$. We then project this history onto the filter bank:
$
h_t \triangleq \Phi^\dagger U_t \in \mathbb{C}^k.
$
The vector $h_t$ is a compressed feature representation of the recent past. A standard online regression algorithm can then be used to map $h_t$ to a prediction of the next observation.

\begin{algorithm}[h!]
\caption{Complex Spectral Filtering}
\label{alg:qsf}
\begin{algorithmic}[1]
\STATE \textbf{Input:} Window $W$, horizon $T$, spectral bound $\beta$, filter size $k$, base online learner $\mathcal{A}$
\STATE \textbf{Initialize:}
\STATE \quad Compute the top $k$ eigenvectors of $Z_W(\beta)$
\STATE \quad Form the filter bank $\Phi \in \mathbb{C}^{W\times k}$
\STATE \quad Initialize $\mathcal{A}$ in dimension $k$
\FOR{$t=1,2,\dots,T$}
    \STATE Observe input history $U_t=[u_t,u_{t-1},\dots,u_{t-W+1}]^\top$ (zero-padded if needed)
    \STATE \textbf{Filter:} Compute $h_t=\Phi^\dagger U_t \in \mathbb{C}^k$
    \STATE \textbf{Predict:} Query $\mathcal{A}$ on $h_t$ to obtain $\hat y_{t+1}$
    \STATE Observe the next outcome $y_{t+1}$
    \STATE \textbf{Update:} Update $\mathcal{A}$ with the pair $(h_t,y_{t+1})$
\ENDFOR
\end{algorithmic}
\end{algorithm}

Complex Spectral Filtering is an \emph{improper} learning method: it does not attempt to identify the hidden transition matrix or recover a minimal realization of the system. Instead, it learns directly in a universal low-dimensional feature space induced by the spectral geometry of the class.

This is particularly useful when the hidden state dimension is large but the response class is spectrally simple. In that regime, the relevant complexity is not the ambient dimension of the latent state but the effective spectral dimension $k^\star$.

\section{Performance Guarantee}

Henceforth we assume the \textit{realizable} setting, where the observations are generated by a generic complex-valued linear dynamical system with dynamics constrained to the spectral sector $\mathbb{C}_\beta$.

\begin{theorem}[Realizable CLDS Guarantee ]
\label{thm:regret}
Let $\{(u_t,y_t)\}_{t=1}^T$ be generated by a noiseless complex linear dynamical system with memory bounded by $W$ and diagonalizable transition matrix $A$ whose spectrum is contained in $\Cbeta$. Assume the inputs are bounded, $\|u_t\|_2 \le 1$ for all $t$.

Suppose Algorithm~\ref{alg:qsf} is run with filter bank size \(k\), and the
base learner is the Vovk--Azoury--Warmuth (VAW) forecaster \cite{azoury2001relative}. Define
$    V_t
    \triangleq
    \lambda I_k + \sum_{s=1}^{t} h_s h_s^\dagger,
    \qquad
    b_{t-1}
    \triangleq
    \sum_{s=1}^{t-1} h_s \overline{y_s}.
$
The prediction is
\begin{equation}
\label{eq:vaw_update}
    \hat y_t
    =
    b_{t-1}^\dagger V_t^{-1} h_t
    =
    \left(\sum_{s=1}^{t-1} y_s h_s^\dagger\right)
    \left(
        \lambda I_k + \sum_{s=1}^{t} h_s h_s^\dagger
    \right)^{-1}
    h_t .
\end{equation}
where $h_t = \Phi^\dagger U_t$ is the filtered history and $\lambda > 0$.
Then the cumulative squared prediction error satisfies
\begin{eqnarray*}
\label{eq:main_regret_bound}
    \sum_{t=1}^{T} |\hat{y}_{t} - y_{t}|^2 & 
    \;\le\;
    \tilde{O}\! \left( k \log T +       T \cdot (\beta W)^4 \cdot
        \exp\!\left(
            -\,c\,\frac{k}{k^\star \log W}
        \right)  \right),
\end{eqnarray*}
where $c>0$ is a universal constant, and the $\tilde{O}(\cdot)$ notation hides polylogarithmic factors and constants depending on the norm of diagonalizing matrices for $A$ and the norms of the system matrices $|A|,|B|,|C|$.
\end{theorem}

The first term in \eqref{eq:main_regret_bound} is the statistical cost of learning a $k$-dimensional linear predictor in the filtered feature space. The second term is the approximation error incurred by truncating the universal filter bank after $k$ directions. By Theorem~\ref{thm:cutoff}, the discarded spectral tail decays exponentially once $k$ passes the effective dimension $k^\star$, so the approximation term becomes negligible with only a logarithmic overspecification beyond the intrinsic complexity of the class.
This yields the following fast-rate corollary.

\begin{corollary}[Fast Rates for Sector-Bounded CLDS]
\label{cor:main}
Under the assumptions of Theorem~\ref{thm:regret}, choosing
$
k
=
\Theta\!\left(
    \frac{\beta}{\pi}W  \log^2 (T \beta W)
\right)
$
gives
\[
\frac{1}{T}\sum_{t=1}^{T} |\hat{y}_{t} - y_{t}|^2
\;\le\;
\tilde{O}\!\left(
    \frac{\beta W \log^2(T \beta W)}{T}
\right).
\]
\end{corollary}

The key qualitative point is that the leading rate depends on the \emph{effective spectral complexity} $\beta W$, rather than on the ambient hidden state dimension $d$.

Due to space constraints, the analysis is deferred to appendices \ref{sec:analysis-main} and \ref{app:hadamard_decay}.

\section{Experimental Validation}

\subsection{Synthetic Sector-Bounded CLDS}

To empirically validate the effective-dimension prediction
$ k^\star = \left\lceil \frac{\beta}{\pi} W \right\rceil$,
we study a family of simulated complex linear dynamical systems whose spectra are constrained to the sector \(\mathbb{C}_{\beta}\). The goal is to test whether learnability is governed by the effective spectral dimension \(k^\star\), rather than by the ambient hidden dimension.

For each trial, we generate a random diagonalizable CLDS with hidden dimension \(d=20\) as follows:
\textbf{Spectral constraint.} The eigenvalues of the transition matrix \(A\) are drawn uniformly from the annular sector
    \[
    \{z \in \mathbb{C} : 0.85 \le |z| \le 0.95,\ |\arg(z)| \le \beta\}.
    \]
 \textbf{Finite-window generation.} We generate observations according to the length-\(W\) convolution
    \[
    y_t = \sum_{\tau=1}^{W} C A^{\tau-1} B u_{t-\tau},
    \]
    with \(W=100\), where the inputs \(u_t\) are i.i.d.\ standard complex Gaussian.\\
 \textbf{Train/test split.} For each system, we generate a trajectory of length \(T=1000\), using the first \(800\) samples for training and the remaining \(200\) samples for evaluation.

We apply Complex Spectral Filtering (Algorithm~\ref{alg:qsf}) with a filter bank built from the leading eigenvectors of \(Z_W(\beta)\). For each filter budget \(K \in \{5,10,\dots,95\}\), we project the lagged input history onto the top \(K\) filters and fit the linear readout with ridge regression (\(\lambda = 10^{-5}\)). We report the test mean-squared error (MSE).

We consider three spectral regimes: narrow sector (\(\beta = 0.2\pi\)), medium sector (\(\beta = 0.5\pi\)), and wide sector (\(\beta = 0.9\pi\)). For each configuration we average over 20 independent trials and report 95\% confidence intervals.

Figure~\ref{fig:clds_phase_transition} shows a clear phase transition: test error stays high until the model capacity \(K\) approaches the predicted effective dimension \(k^\star\), and then rapidly drops to the numerical floor. For the low-bandwidth regime \(\beta = 0.2\pi\), near-perfect prediction is achieved with roughly \(K=20\) filters even though the raw window length is \(W=100\). This matches the compression effect predicted by Theorem~\ref{thm:cutoff}.

\begin{figure}[htbp]
    \centering
    \begin{subfigure}[t]{0.32\textwidth}
        \centering
        \includegraphics[width=\linewidth]{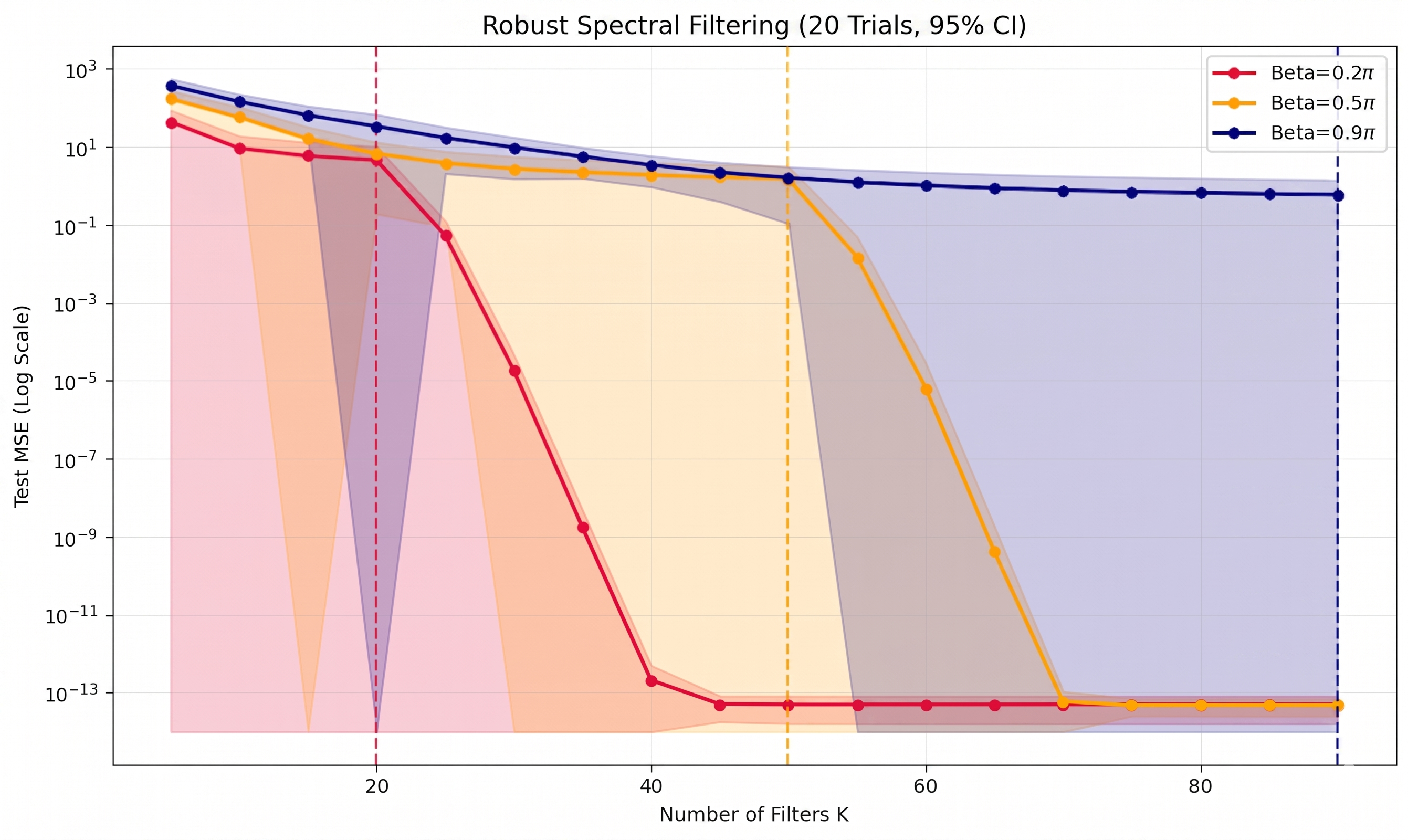}
        \caption{\textbf{Varying spectral regimes.} The transition near \(k^\star\) matches the effective-dimension prediction.}
        \label{fig:clds_phase_transition}
    \end{subfigure}\hfill
    \begin{subfigure}[t]{0.32\textwidth}
        \centering
        \includegraphics[width=\linewidth]{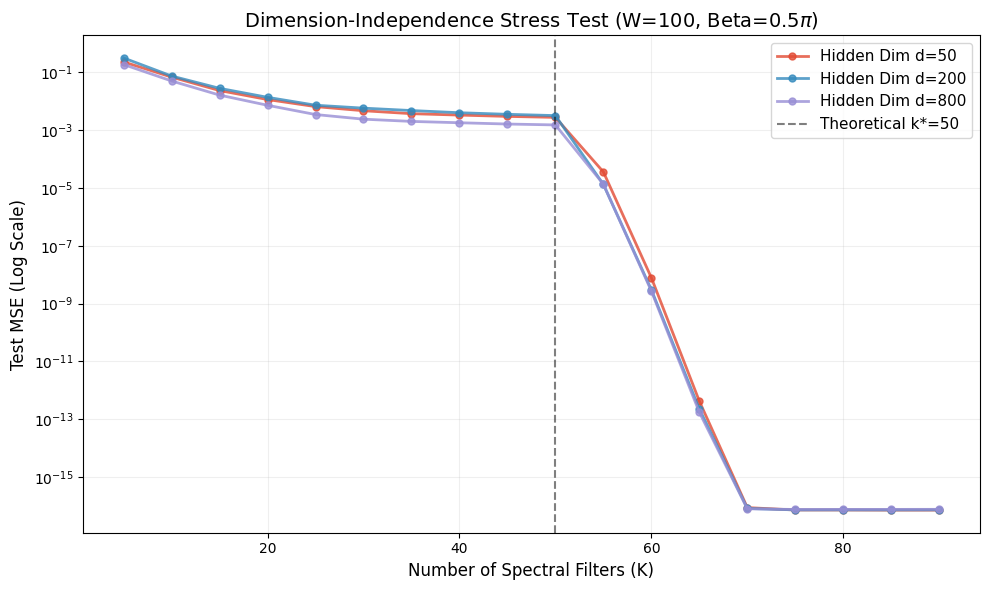}
        \caption{\textbf{Varying hidden dimension.} The curves coincide, showing independence from ambient dimension \(d\).}
        \label{fig:dim_test}
    \end{subfigure} \hfill
    \begin{subfigure}[t]{0.32\textwidth}
        \centering
        \includegraphics[width=\linewidth]{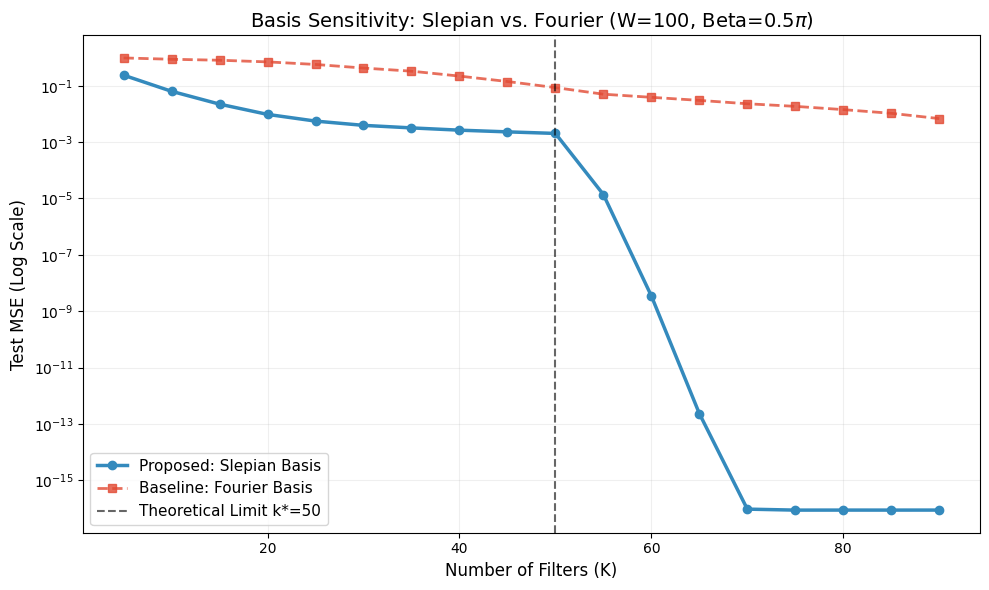}
        \caption{\textbf{Basis sensitivity: Slepian vs. Fourier.}  Slepian basis reaches numerical floor at cutoff,  Fourier baseline converges more slowly. } 
        \label{fig:basis_sensitivity}
    \end{subfigure}\hfill
    \caption{\textbf{Synthetic sector-bounded CLDS Experiments.} Test MSE of Complex Spectral Filtering on random CLDS with window size \(W=100\). \textbf{(a)} We vary the filter budget \(K\) across three spectral regimes (\(\beta=0.2\pi, 0.5\pi, 0.9\pi\)). Solid lines show the mean over 20 trials with 95\% confidence intervals. Vertical dashed lines mark the predicted effective dimension \(k^\star=\left\lceil \frac{\beta}{\pi}W \right\rceil\). \textbf{(b)} Test MSE for varying hidden dimensions \(d \in \{50,200,800\}\) with fixed \(\beta=0.5\pi\). The vertical dashed line marks \(k^\star = 50\).
    \textbf{(c)}  Test MSE comparison between Complex Spectral Filtering using the proposed Slepian basis and a baseline using low-frequency DFT modes. The experiment uses a random CLDS with \(W=100\) and \(\beta=0.5\pi\). The vertical dashed line marks the predicted effective dimension \(k^\star = 50\). 
    \label{fig:combined_experiments}}
\end{figure}


\subsection{Ablation Study: Effect of Hidden Dimension}

A central claim of our framework is that prediction difficulty is controlled by spectral complexity rather than by ambient state dimension. To test this, we fix \(W=100\) and \(\beta=0.5\pi\) and vary the hidden dimension \(d \in \{50,200,800\}\).

All other aspects of the data-generating process remain unchanged. Figure~\ref{fig:dim_test} shows that the learning curves are nearly indistinguishable across all three dimensions. The location of the transition remains stable and continues to occur near \(k^\star = 50\), which is consistent with the dimension-free dependence highlighted in Theorem~\ref{thm:regret}. Even though the ambient transition matrix becomes much larger as \(d\) increases, the number of spectral features needed for accurate prediction is essentially unchanged.


\subsection{Ablation Study: Slepian Basis vs.\ Fourier Basis}

A core theoretical claim of this paper is that the Slepian/DPSS basis is the right finite-window representation for sector-bounded CLDS. To test this, we compare Complex Spectral Filtering with a baseline that replaces the Slepian filter bank by the top-\(K\) low-frequency DFT modes.

We generate trajectories from a random complex-valued CLDS with hidden dimension \(d=100\), window \(W=100\), and sector width \(\beta=0.5\pi\). Both methods use the same training and evaluation protocol; only the choice of basis changes.

Figure~\ref{fig:basis_sensitivity} shows that the Slepian basis exhibits a much sharper transition and reaches the numerical floor near the predicted cutoff \(k^\star = 50\). The Fourier basis improves more slowly and leaves a pronounced tail of error. This is consistent with the fact that DPSS vectors are optimally concentrated on finite windows, whereas DFT modes suffer from spectral leakage in the truncated setting.


\section{Conclusion}

We introduced a spectral filtering framework for learning complex-valued linear dynamical systems whose spectra are confined to a sector of the unit disk. Rather than identifying the latent state-space realization, our method learns directly from finite-context input histories using a universal filter bank determined only by the spectral sector width $\beta$ and context window $W$. The central message is that the predictive complexity of this class is governed not by the ambient hidden dimension, but by an effective spectral dimension of order $\beta W$.

Our analysis connects this phenomenon to classical time--frequency concentration: sector-bounded complex dynamics induce finite-window impulse responses that are well-approximated by a low-dimensional Slepian-type subspace. This yields dimension-free regret guarantees for online prediction, where the statistical complexity scales with the effective spectral dimension rather than the full memory length or latent dimension. Experiments on synthetic sector-bounded CLDS support the predicted phase transition near the Shannon-number cutoff and demonstrate the advantage of the Slepian basis over standard Fourier features.

More broadly, these results suggest that spectral localization is a useful structural principle for learning oscillatory and long-memory dynamical systems. Future directions include extending the analysis to noisy observations, non-diagonalizable or non-normal dynamics, adaptive estimation of the sector width $\beta$, and deeper exploration of applications to structured state space models and predictive quantum dynamics.

\bibliographystyle{plain}
\bibliography{quantum}

\appendix
\onecolumn

\section{Lower bounds for learning an LDS} \label{appendix:lower_bound}

In this section we establish a simple lower bound for learning a linear dynamical system with hidden dimension $d$ (which can be thought of as $W$ in the earlier notation of this paper), if we don't restrict $\beta$. This shows that our refined bounds require a bound on $\beta$.

While we construct a specific hard instance to prove the $\Omega(d)$ bound, the paper of \cite{franccois2025uncertainty} provides a more refined lower bound, but one based on stricter assumptions. Specifically, they show a granular lower bound of $1 - S/K$ for the task of learning a time-shift of $K$ with state dimension $S$ on white noise, further illustrating the necessity of spectral constraints (small $\beta$) for efficient learning. Their lower bound, however, restricts the model class to linear recurrent filters, whereas our bound applies to any learning algorithm.

\begin{theorem}
\label{thm:lds_lower_bound}
For any online learning algorithm $\mathcal{A}$, there exists a linear dynamical system $(A, B, C, D)$ with hidden state dimension $d$, and a sequence of observations, such that the expected cumulative squared loss of $\mathcal{A}$ over $T \ge d$ rounds is at least $\Omega(d)$.
\end{theorem}

\begin{proof}
We construct a specific autonomous LDS where predicting the observations is information-theoretically equivalent to predicting a sequence of independent fair coin flips.

Consider an autonomous LDS defined by state dimension $d$, input dimension $0$ (so $B=0, D=0$), and output dimension $1$. The dynamics are given by:
\begin{align*}
    h_{t+1} &= A h_t \\
    y_t &= C h_t
\end{align*}
Let $A$ be the cyclic permutation matrix that shifts vector components one position to the left:
$$
A = \begin{pmatrix}
0 & 1 & 0 & \cdots & 0 \\
0 & 0 & 1 & \cdots & 0 \\
\vdots & \vdots & \vdots & \ddots & \vdots \\
0 & 0 & 0 & \cdots & 1 \\
1 & 0 & 0 & \cdots & 0
\end{pmatrix}
$$
Let $C$ be the projection onto the first coordinate:
$$
C = \begin{pmatrix} 1 & 0 & \cdots & 0 \end{pmatrix}
$$
The initial state $h_0 \in \{-1, 1\}^d$ is chosen uniformly at random. This means each component $(h_0)_i$ is an independent Rademacher random variable, with $\mathbb{P}((h_0)_i = 1) = \mathbb{P}((h_0)_i = -1) = 1/2$.

At any time step $t \ge 0$, the state is $h_t = A^t h_0$ and the observation is $y_t = C A^t h_0$. Due to the structure of $A$ and $C$, for the first $d$ rounds $t \in \{0, 1, \dots, d-1\}$, we have precisely:
$$
y_t = (h_0)_{t+1}
$$
That is, the observation at time $t$ reveals exactly the $(t+1)$-th component of the initial unknown state vector $h_0$.

For any deterministic online learning algorithm $\mathcal{A}$, its prediction $\hat{y}_t$ at time $t$ is a function only of past observations:
$$
\hat{y}_t = f_t(y_0, y_1, \dots, y_{t-1})
$$
Substituting our relation $y_i = (h_0)_{i+1}$, this means $\hat{y}_t$ depends only on the first $t$ components of $h_0$:
$$
\hat{y}_t = f_t((h_0)_1, (h_0)_2, \dots, (h_0)_t)
$$
However, the true label $y_t$ is $(h_0)_{t+1}$, which is statistically independent of the variables $(h_0)_1, \dots, (h_0)_t$. Consequently, $\mathbb{E}[y_t | y_0, \dots, y_{t-1}] = \mathbb{E}[(h_0)_{t+1}] = 0$.

The expected squared loss at any round $t < d$ is bounded from below:
\begin{align*}
\mathbb{E}[(y_t - \hat{y}_t)^2] &= \mathbb{E}[y_t^2 - 2 y_t \hat{y}_t + \hat{y}_t^2] \\
&= \mathbb{E}[y_t^2] - 2\mathbb{E}[y_t]\mathbb{E}[\hat{y}_t] + \mathbb{E}[\hat{y}_t^2] \quad \text{(by independence)} \\
&= 1 - 0 + \mathbb{E}[\hat{y}_t^2] \\
&\ge 1
\end{align*}
Summing this expected loss over the first $d$ rounds yields a total expected loss for the algorithm of at least $d$.
Since the system is deterministic given $h_0$, the best expert in hindsight (knowing $h_0$) achieves a total loss of $0$.
Therefore, the expected regret is at least $d - 0 = d$.
\end{proof}

\section{Analysis of Main Theorem} \label{sec:analysis-main}

The proof of our theorem makes use of the following theorem from \cite{marsden2025universal}. 
\begin{theorem}[Proof of Lemma E.4\footnote{The result comes from combining Eq. 11 with Lemma A.1} \cite{marsden2025universal}]
Let $\{\phi_i\}_{i=1}^W$ be the eigenvectors of $Z_W(\beta)$ sorted by eigenvalue $\sigma_1 \ge \sigma_2 \ge \dots \ge \sigma_W$. For any $z \in \mathbb{C}_\beta$, the projection of its monomial trajectory onto the $i$-th basis vector decays with the corresponding eigenvalue $\sigma_i$ as follows,
    \begin{equation}
\label{eqn:bound1}
   \max_{z \in \mathbb{C}_\beta} |\phi_i^{\dagger} \mu_W(z)|^2 \leq \left( 24 \cdot (12)^2 \beta^7 W^4 \sigma_i \right)^{1/3}.
\end{equation}
\end{theorem}

Our analysis departs from that of \cite{marsden2025universal} in that we prove and exploit much stronger spectral properties of the Spectral Information Matrix. Indeed, applying Theorem~\ref{thm:cutoff} to the above gives our final result. 
\begin{lemma}[Spectral Decay]
\label{lem:uniform_approx}
Let $\{\phi_i\}_{i=1}^W$ be the eigenvectors of $Z_W(\beta)$ sorted by eigenvalues
$\sigma_1 \ge \sigma_2 \ge \dots \ge \sigma_W$.
Then for any $z \in \mathbb{C}_\beta$ and any $i > k^*$,
\[
    \max_{z \in \mathbb{C}_\beta}
    |\phi_i^{\dagger}\mu_W(z)|^2
    \le
    C \beta^{7/3} W^{4/3}
    \exp\!\left(
        -\,c\,\frac{i}{k^* \log W}
    \right),
\]
for universal constants $C,c>0$.
\end{lemma}

\begin{proof}[Detailed Proof of Theorem \ref{thm:regret}]
We decompose the total cumulative error into \textit{estimation error} (regret) and \textit{approximation error} (bias).
Let $w^* = \text{argmin}_{w \in \mathbb{C}^{k}} \sum_{t=1}^{T} |w^* h_{t} - y_{t}|^2$ be the best offline predictor in our chosen subspace $\Phi$.
The total error is:
\begin{align*}
    \sum_{t=1}^{T} |\hat{y}_{t} - y_{t}|^2 & =  \underbrace{\left( \sum_{t=1}^{T} |\hat{y}_{t} - y_{t}|^2 - \sum_{t=1}^{T} |w^{*\top} h_{t} - y_{t}|^2 \right)}_{\text{Estimation Error}} \\ 
    & + \underbrace{\sum_{t=1}^{T} |w^\top h_{t} - y_{t}|^2}_{\text{Approximation Error}}
\end{align*}

\paragraph{1. Estimation Error (Special Case: VAW):}
We analyze the specific case where $\mathcal{A}$ is the Vovk-Azoury-Warmuth (VAW) forecaster \cite{azoury2001relative} (also known as Forward-Step Ridge Regression). Following the analysis of Follow-The-Regularized-Leader (FTRL) with quadratic regularization (e.g., Orabona, 2019), the cumulative regret is bounded by the log-determinant of the data covariance matrix.
For the square loss, the regret satisfies:
\begin{equation}
    \mathcal{R}_T(w^*) \le \lambda \|w^*\|^2 + \sum_{t=1}^T \log \left( 1 + h_t^* V_{t-1}^{-1} h_t \right),
\end{equation}
where $V_{t-1} = \lambda I_k + \sum_{s=1}^{t-1} h_s h_s^{\dagger}$ is the regularized covariance matrix.
Using the Elliptical Potential Lemma, the sum term is exactly the log-ratio of determinants:
\[ \sum_{t=1}^T \log(1 + h_t^* V_{t-1}^{-1} h_t) = \log \frac{\det(V_T)}{\det(\lambda I)}. \]
Assuming bounded inputs $\|h_t\|^2 \le L^2$ within the subspace, the trace of the covariance matrix is bounded by $\text{Tr}(V_T) \le k \lambda + T L^2$. By the AM-GM inequality, the determinant is maximized when the energy is distributed equally among the $k$ eigenvalues:
\[ \det(V_T) \le \left( \frac{\text{Tr}(V_T)}{k} \right)^{k} \le \left( \lambda + \frac{T L^2}{k} \right)^{k}. \]
Taking the logarithm yields the dimension-dependent bound:
\[ \log \det(V_T) - k \log \lambda \le k \log \left( 1 + \frac{T L^2}{\lambda k} \right). \]
Thus, the estimation error scales linearly with the filter bank size $k$ (not $W$):
\[ \text{Estimation Error} \le \tilde{O}(k \log T). \]

\paragraph{2. Approximation Error:}
We bound the error of the best linear predictor $w^*$ in our basis.
Assuming the system is a linear time-invariant (LTI) system driven by inputs $u_t$ with memory $W$, the true observation is 
$$y_{t+1} = \sum_{\tau=1}^W C A^{\tau-1} B u_{t-\tau+1} = \sum_{\tau=1}^W \bar{C} D^{\tau-1} \bar{B} u_{t-\tau+1} , $$
where $A$ is diagonalized into $A = P D P^{-1}$, for a diagonal matrix $D$, and $P,P^{-1}$ are absorbed into $\bar{C} = C P , \bar{B} = P^{-1} B$.
Let the $d$ eigenvalues of $A$ be $\{z_j\} \subset \mathbb{C}_\beta$, then we can write the true system as a linear combination of eigenvector products as:
\[ y_{t+1} = \sum_{j=1}^d c_j \cdot \mu_W(z_j)^\top U_t , \]
where recall that $U_t = [u_t, u_{t-1}, \dots, u_{t-W+1}]$ (padded if necessary for $t < W$), $\mu_W(z) = [1, z, z^2, \dots, z^{W-1}]^\top$, and $c_j$ is given by $c_j = \bar{C}_j^\top  \bar{B}_{j}$.

We choose the ideal weights $w^* = \sum_{j=1}^d {c}_j (\Phi^{\dagger} \mu_W(z_j))$ to match the projection of these true parameters onto our basis $\Phi$.
The approximation error at time $t$ is:
\begin{align*}
    e_t &= w^{*\top} h_t - y_{t+1} = - \sum_{j=1}^d c_j \sum_{i=k+1}^W (\phi_i^{\dagger} \mu_W(z_j)) (\phi_i^\top U_t)
\end{align*}
Applying Cauchy-Schwarz and assuming bounded inputs, i.e. $\|u_t\|_2 \leq 1$ (so that $\|U_t\|^2 \le W$):
\begin{align*}
    |e_t|^2 &\le \left( \sum_{j=1}^d |c_j| \sqrt{\sum_{i=k+1}^W |\phi_i^{\dagger} \mu_W(z_j)|^2} \sqrt{\sum_{i=k+1}^W |\phi_i^\top U_t|^2} \right)^2 \\
    &\le \left( \sum_{j=1}^d |c_j| \sqrt{\sum_{i=k+1}^W \max_{z \in \mathbb{C}_\beta} |\phi_i^{\dagger} \mu_W(z)|^2} \cdot \|U_t\|_2 \right)^2
\end{align*}
Using Lemma \ref{lem:uniform_approx}, each term in the sum decays exponentially and the entire sum can be bounded by a multiple of its first term at $i = k+1$.
\begin{align*}
|e_t|^2 & \le \left( \|c\|_1 \cdot \sqrt{W \cdot  \left(C \beta^{7/3} W^{4/3} \exp(-ck/(k^*\log W)) \right)  } \right)^2 \\
& \leq C \|c\|_1^2 \beta^{7/3} W^{7/3} \exp(-ck/(k^* \log W)).
\end{align*}
Notice that each term $c_i$ can be bounded as $|c_i| = |\bar{C}_i||\bar{B}_i| \leq |P||P^{-1}||C_i||B_i| $. This is a constant which depends on the condition number of $A$, as well as the norms of the matrices $B,C$. Since eventually it will appear logarithmically in the choice of $k$, we consider it a constant. Moreover, since the powers of $\beta$ and $W$ will also appear only logarithmically we simplify them with an upper bound of $7/3 \leq 3$. Thus, 
\[ |e_t|^2 \le C W^3 \beta^3 \exp(-\frac{c k}{ k^* \log W}) . \]

Summing this constant error bound over all $T$ time steps yields the final approximation term.
Combining these two terms yields the final generative bound.
\end{proof}

\section{Proof of Theorem~\ref{thm:cutoff}: Spectral Decay via Hadamard Products}
\label{app:hadamard_decay}

This appendix gives a proof of Theorem~\ref{thm:cutoff},
establishing stretched-exponential spectral decay governed by the
effective dimension $k^*$. The matrix is decomposable as a Hadamard product of a Hankel matrix and the so-called Slepian matrix. The theorem combines the work of \cite{beckermann2017singular} and \cite{hazan2017learning} which show the decay of the Hankel matrix and the work of Slepian which characterizes the spectrum of the DPSS matrix. This combination uses a novel and critical lemma on the spectrum of the Hadamard product of two PSD matrices to get the final result.  

\subsection{A General Hadamard Eigenvalue Tail Bound}

\begin{lemma}[Generalized Hadamard Eigenvalue Decay]
\label{lem:hadamard_decay}
Let $A, B \in \mathbb{C}^{n \times n}$ be Hermitian positive semi-definite matrices with
eigenvalues $\lambda_1 \ge \lambda_2 \ge \dots \ge 0$.
For any integers $k, h \ge 1$, the $(kh+1)$-th eigenvalue of their Hadamard product
$A \circ B$ satisfies
\begin{equation}
\lambda_{kh+1}(A \circ B)
\;\le\;
\lambda_1(B)\,\mathrm{Tr}(A_{>k})
\;+\;
\lambda_1(A)\,\mathrm{Tr}(B_{>h})
\;+\;
\mathrm{Tr}(A_{>k})\,\mathrm{Tr}(B_{>h}),
\end{equation}
where $\mathrm{Tr}(M_{>r}) := \sum_{j=r+1}^n \lambda_j(M)$ denotes the tail sum.
\end{lemma}

\begin{proof}
Let
\[
A = \sum_j \mu_j u_j u_j^\dagger,
\qquad
B = \sum_i \zeta_i w_i w_i^\dagger
\]
be eigendecompositions.
Decompose
\[
A = A_{\le k} + A_{>k},
\qquad
B = B_{\le h} + B_{>h}.
\]
Then
\[
A \circ B
=
\underbrace{A_{\le k} \circ B_{\le h}}_{M_1}
+
\underbrace{A_{>k} \circ B_{\le h}}_{M_2}
+
\underbrace{A_{\le k} \circ B_{>h}}_{M_3}
+
\underbrace{A_{>k} \circ B_{>h}}_{M_4}.
\]

\paragraph{Rank bound.}
\[
M_1
=
\sum_{j=1}^k \sum_{i=1}^h
\mu_j \zeta_i
(u_j \circ w_i)(u_j \circ w_i)^\dagger,
\]
hence $\mathrm{rank}(M_1) \le kh$ and $\lambda_{kh+1}(M_1)=0$.

\paragraph{Spectral bounds.}
Recall that Weyl's inequality states that
\begin{equation*}
    \lambda_{i+j-1}(A + B) \leq \lambda_i(A) + \lambda_j(B).
\end{equation*}
Therefore we have
\[
\lambda_{kh+1}(A \circ B)
\le
\lambda_{kh+1}(M_1) + \lambda_1(M_2+M_3+M_4) \le \lambda_1(M_2) + \lambda_1(M_3) + \lambda_1(M_4).
\]
In general it holds that for any PSD matrix P, non-negative $\rho_i$, and orthonormal vectors $v_i$
\begin{equation*}
    \lambda_1 \left( \sum_{i=1}^n \rho_i v_i v_i^{\dagger} \circ P \right) \leq \sum_{i=1}^n \rho_i \lambda_1(P).
\end{equation*}
Indeed, by Weyl's inequality we have
\begin{equation*}
     \lambda_1 \left( \sum_{i=1}^n \rho_i v_i v_i^{\dagger} \circ P \right) \leq \sum_{i=1}^n \lambda_1 \left( \rho_i v_i v_i^{\dagger} \circ P  \right) = \sum_{i=1}^n \rho_i \lambda_1 \left(  v_i v_i^{\dagger} \circ P  \right). 
\end{equation*}
For any PSD matrix $P$ and vector $v$ such that $\norm{v}_2 \leq 1$ it is true that $\lambda_1(vv^\dagger \circ P) \le \lambda_1(P)$. Indeed, if $D_v= \textrm{diag}(v)$ we can write $vv^\dagger \circ P$ as $D_v P D_v^{\dagger}$. Then
\begin{equation*}
    \lambda_1(D_v P D_v^{\dagger}) = \max_{x: \norm{x}_2 \leq 1} x^{\dagger} \left( D_v P D_v^{\dagger} \right)x  = \max_{y: \norm{D_v^{\dagger} y} \leq 1} y^{\dagger} P y \leq  \max_{y: \norm{ y} \leq 1} y^{\dagger} P y = \lambda_1(P),
\end{equation*}
where the last inequality holds since $\norm{D_v}_2 \leq 1$ and therefore whenever $\norm{y}_2\leq 1$ implies $\norm{D_v y}_2 \leq 1$. Therefore we have,
\begin{align*}
\lambda_1(M_2)
&\le
\mathrm{Tr}(A_{>k})\,\lambda_1(B),\\
\lambda_1(M_3)
&\le
\mathrm{Tr}(B_{>h})\,\lambda_1(A),\\
\lambda_1(M_4)
&\le
\mathrm{Tr}(A_{>k})\,\mathrm{Tr}(B_{>h}).
\end{align*}
Summing completes the proof.
\end{proof}

\subsection{The Hankel Matrix Component}
We borrow Lemma E.1 from \cite{hazan2017learning} which is a result that largely follows from \cite{beckermann2017singular}.
    \begin{lemma}[Lemma E.1 \cite{hazan2017learning}]
    \label{lemma:hankel_decay}
        Let $\lambda_j$ be the $j$-th top eigenvalue of $A_T$. Then for all $T \geq 10$,
        \begin{equation}
            \lambda_j \leq \min \left( \frac{3}{4}, K \cdot c^{-j/ \log T} \right)
        \end{equation}
    \end{lemma}
    
\subsection{The Prolate (DPSS) Component}

This subsection states the spectral properties of the discrete prolate
spheroidal sequence (DPSS) kernel, also called the Slepian kernel, following
the classical work of \cite{slepian1978prolate} and the non-asymptotic
eigenvalue bounds of \cite{karnik2021improved}.

\begin{lemma}[Identification with Slepian's Prolate Matrix from \cite{slepian1978prolate}]
\label{lem:prolate_identification}
Let $B\in\mathbb{C}^{W\times W}$ be the Toeplitz matrix with entries
\[
B_{jk} \;=\; 2\beta\,\sinc\!\big((j-k)\beta\big),
\qquad j,k\in\{0,1,\dots,W-1\},
\]
where $\sinc(x)=\sin(x)/x$ with $\sinc(0)=1$.
Define $\beta_{\mathrm{slep}}:=\beta/(2\pi)$ and let $p(W,\beta_{\mathrm{slep}})\in\mathbb{R}^{W\times W}$ be
Slepian's prolate matrix with entries
\[
p(W,\beta_{\mathrm{slep}})_{jk}
\;:=\;
\frac{\sin\!\big(2\pi \beta_{\mathrm{slep}}(j-k)\big)}{\pi(j-k)},
\quad (j\neq k),
\qquad
p(W,\beta_{\mathrm{slep}})_{jj}:=2\beta_{\mathrm{slep}}.
\]
Then
\[
B \;=\; 2\pi\, p(W,\beta_{\mathrm{slep}}).
\]
In particular, the eigenvalues of $B$ are exactly $2\pi$ times the eigenvalues
$\lambda_k(W,\beta_{\mathrm{slep}})$ studied by Slepian, and the associated eigenvectors
are the discrete prolate spheroidal sequences (DPSS).
Moreover, the Shannon number satisfies
\[
2W\,\beta_{\mathrm{slep}} \;=\; \frac{\beta}{\pi}W \;=\; k^* .
\]
\end{lemma}

\begin{proof}
Using $\sinc(x)=\sin(x)/x$, for $j\neq k$ we have
\[
B_{jk}
=2\beta\cdot \frac{\sin\!\big((j-k)\beta\big)}{(j-k)\beta}
=
2\cdot \frac{\sin\!\big((j-k)\beta\big)}{(j-k)}
=
2\pi\cdot \frac{\sin\!\big(2\pi \beta_{\mathrm{slep}}(j-k)\big)}{\pi(j-k)}
=
2\pi\, p(W,\beta_{\mathrm{slep}})_{jk},
\]
where we used $\beta=2\pi \beta_{\mathrm{slep}}$.
On the diagonal, $B_{jj}=2\beta$ and
$2\pi p_{jj}=2\pi(2\beta_{\mathrm{slep}})=4\pi \beta_{\mathrm{slep}}=2\beta$.
The Shannon number identity follows immediately.
\end{proof}


\begin{theorem}[DPSS Spectral Concentration and Cutoff
{\cite{slepian1978prolate,karnik2021improved}}]
\label{thm:dpss_cutoff}
Let $B$ be as in Lemma~\ref{lem:prolate_identification}, and let
\[
    k^\star := \frac{\beta}{\pi}W.
\]
Then the spectrum of $B$ exhibits a sharp concentration phenomenon:
there exist constants $\delta_\beta,a_\beta,C_\beta>0$, depending only on
the bandwidth fraction $\beta_{\mathrm{slep}}$, such that
\begin{itemize}
\item (\emph{Passband}) For $k \le (1-\delta_\beta)k^\star$,
\[
    \lambda_k(B)\ge 2\pi\bigl(1-C_\beta e^{-a_\beta W}\bigr).
\]
\item (\emph{Stopband}) For $k \ge (1+\delta_\beta)k^\star$,
\[
    \lambda_k(B)\le 2\pi C_\beta e^{-a_\beta W}.
\]
\end{itemize}
Consequently, if
\[
    k^\star_+
    :=
    \min\left\{
        W,\,
        \left\lceil (1+\delta_\beta)k^\star \right\rceil
    \right\},
\]
then there exist constants $C'_\beta,a'_\beta>0$ such that
\[
    \sum_{j>k^\star_+}\lambda_j(B)
    \le
    C'_\beta e^{-a'_\beta W}.
\]
\end{theorem}

\subsection{Spectral Decay for Slepian--Hankel Products}

We now combine the Hadamard tail bound
(Lemma~\ref{lem:hadamard_decay}) with the Hankel concentration
(Lemma~\ref{lemma:hankel_decay}) and the DPSS concentration
(Theorem~\ref{thm:dpss_cutoff}) to obtain the spectral decay of the
Spectral Information Matrix $Z_W(\beta)=A\circ B$. Theorem~\ref{thm:cutoff}
follows directly from the theorem below, since
$k^\star_+=\Theta_\beta(k^\star)$.

\begin{theorem}[Slepian--Hadamard Spectral Decay]
Let $A,B \in \mathbb{C}^{W\times W}$ be Hermitian PSD matrices satisfying:
\begin{enumerate}
    \item \textbf{Hankel decay:} there exist constants $C_A,c_A>0$ such that
    \[
    \lambda_j(A)
    \le
    C_A \exp\!\left(-c_A \frac{j}{\log W}\right)
    \qquad\text{for all } j\ge 1.
    \]

    \item \textbf{Slepian structure:} there exist constants $C_B,c_B>0$ and
    a cutoff $k^\star_+=\Theta_\beta(k^\star)$ such that
    \[
    \lambda_1(B)\le C_B,
    \qquad
    \sum_{j>k^\star_+}\lambda_j(B)\le C_B e^{-c_B W}.
    \]
\end{enumerate}

Then there exist constants $K_1,K_2,c>0$ such that for all $s\ge 1$,
\[
\lambda_{s k^\star_+ +1}(A\circ B)
\;\le\;
K_1 \log W \cdot
\exp\!\left(
    -\,c\,\frac{s}{\log W}
\right)
+
K_2 e^{-c_B W},
\]
with the convention that $\lambda_j=0$ for $j>W$. In particular, after
adjusting constants, for every $k\ge k^\star$,
\[
\lambda_k(A\circ B)
\;\le\;
K_1 \log W \cdot
\exp\!\left(
    -\,c\,\frac{k}{k^\star \log W}
\right)
+
K_2 e^{-c_B W}.
\]
\end{theorem}

\begin{proof}
Apply Lemma~\ref{lem:hadamard_decay} with cutoffs $k=s$ and
$h=k^\star_+$. This gives
\[
\lambda_{s k^\star_+ +1}(A\circ B)
\le
\lambda_1(B)\operatorname{Tr}(A_{>s})
+
\lambda_1(A)\operatorname{Tr}(B_{>k^\star_+})
+
\operatorname{Tr}(A_{>s})\operatorname{Tr}(B_{>k^\star_+}).
\]

By the Hankel decay assumption,
\[
\operatorname{Tr}(A_{>s})
\le
\sum_{j>s}
C_A \exp\!\left(-c_A \frac{j}{\log W}\right)
\le
K_A \log W \cdot
\exp\!\left(
    -c_A\frac{s}{\log W}
\right),
\]
where the last inequality follows by summing a geometric series. Also
$\lambda_1(A)\le C_A$. By the Slepian tail assumption,
\[
\operatorname{Tr}(B_{>k^\star_+})
\le
C_B e^{-c_B W}.
\]
Substituting these bounds yields
\[
\lambda_{s k^\star_+ +1}(A\circ B)
\le
K_1 \log W \cdot
\exp\!\left(
    -c\frac{s}{\log W}
\right)
+
K_2 e^{-c_B W},
\]
after absorbing the product term into the first term.

For the second display, use $k^\star_+=\Theta_\beta(k^\star)$ and monotonicity
of the eigenvalues. Specifically, for $k\ge k^\star_+$ take
$s=\lfloor k/k^\star_+\rfloor$ and adjust constants. For the remaining range
$k^\star\le k<k^\star_+$, the ratio $k/k^\star=O_\beta(1)$, so the exponential
factor is bounded below by a constant; the claim then follows from the trivial
bound $\lambda_k(A\circ B)\le \lambda_1(A\circ B)\le C\log W$, again after
adjusting constants.
\end{proof}

\end{document}